\documentclass[11pt,reqno,final]{amsart}

\usepackage{amsmath, amssymb, epsfig}
\usepackage{mathrsfs}
\usepackage{color}
\usepackage{alg}
\usepackage{url}
\usepackage[font=small,margin=10pt,labelfont={bf},labelsep={space}]{caption}
\usepackage{subfig}
\usepackage{wrapfig}

\usepackage[scaled]{helvet} 

\usepackage{fourier} 
\usepackage{bm}

\makeatletter
\def\algspacing{\alg@unmargin}
\makeatother

\usepackage[margin=1in]{geometry}

\newlength{\algorithmwidth}
\theoremstyle{plain}
\newtheorem{theorem}{Theorem}[section]
\newtheorem{proposition}[theorem]{Proposition}
\newtheorem{corollary}[theorem]{Corollary}
\newtheorem{lemma}[theorem]{Lemma}

\theoremstyle{definition}
\newtheorem{definition}[theorem]{Definition}

\theoremstyle{remark}

\numberwithin{theorem}{section}
\numberwithin{equation}{section}

\newcommand{\<}{\left\langle}
\renewcommand{\>}{\right\rangle}

\DeclareMathOperator*{\supp}{supp}

\DeclareMathOperator*{\argmin}{arg min}

\newcommand{\defby}{\overset{\mathrm{\scriptscriptstyle{def}}}{=}}
\def \unrecoverable {localization factor}

\def \C {\mathbb{C}}
\def \R {\mathbb{R}}
\def \W {\mtx{W}}
\def \P {\mathbb{P}}
\def \E {\mathbb{E}}
\def \NN {\mathbb{N}}
\def \N {\mathcal{N}}
\newcommand{\e}{\mathrm{e}}

\newcommand{\vct}[1]{\bm{#1}}
\newcommand{\mtx}[1]{\bm{#1}}

\def \A {\mtx{A}}
\def \Atilde {\mtx{\tilde{A}}}

\def \rtilde {\vct{r}}
\def \D {\mtx{D}}
\def \F {\mtx{F}}
\def \H {\mtx{H}}
\def \tH {\tilde{\mtx{H}}}

\def \B {\mtx{B}}

\def \b {\vct{b}}
\def \r {\vct{r}}
\def \e {\vct{e}}
\def \d {\vct{d}}
\def \x {\vct{x}}
\def \y {\vct{y}}
\def \z {\vct{z}}
\def \f {\vct{f}}
\def \u {\vct{u}}
\def \w {\vct{w}}

\newcommand{\vertiii}[1]{{\left\vert\kern-0.25ex\left\vert\kern-0.25ex\left\vert #1 
    \right\vert\kern-0.25ex\right\vert\kern-0.25ex\right\vert}}
\def \unrec {{\eta}}
\def \Dsn {\mathcal{D}}
\def \DsnOne {\mathcal{D}_{\text{aux}}}
\def \BoneS {\mathcal{B}_{1}^{\Lambda}}

\newcommand{\tripnorm}[1]{\vertiii{#1}_s}
\newcommand{\Xnorm}[1]{\norm{#1}_X}

\newcommand{\ip}[2]{\left\langle {#1}, \, {#2} \right\rangle}
\newcommand{\norm}[1]{\lVert {#1} \rVert}

\newcommand{\enorm}[1]{\norm{#1}_2}
\newcommand{\enormsq}[1]{\enorm{#1}^2}

\newcommand{\Id}{\mathbf{I}}

\begin{document}

\title{Compressive Sensing with Redundant Dictionaries and Structured Measurements} 
\author{Felix Krahmer, Deanna Needell, and Rachel Ward}
       
\date{\today}

\begin{abstract}
Consider the problem of recovering an unknown signal from undersampled measurements, given the knowledge that the signal has a sparse representation in a specified dictionary $\D$.  This problem is now understood to be well-posed and efficiently solvable under suitable assumptions on the measurements and dictionary, if the number of measurements scales roughly with the sparsity level.  One sufficient condition for such is the $\D$-restricted isometry property ($\D$-RIP), which asks that the sampling matrix approximately preserve the norm of all signals which are sufficiently sparse in $\D$.  While many classes of random matrices are known to satisfy such conditions, such matrices are not representative of the structural constraints imposed by practical sensing systems.   We close this gap in the theory by demonstrating that one can subsample a fixed orthogonal matrix in such a way that the $\D$-RIP will hold, provided this basis is sufficiently incoherent with the sparsifying dictionary $\D$.   We also extend this analysis to allow for weighted sparse expansions.   Consequently, we arrive at compressive sensing recovery guarantees for structured measurements and redundant dictionaries, opening the door to a wide array of practical applications. 
\end{abstract}

\maketitle

\section{Introduction}
\subsection{Compressive Sensing}
The compressive sensing paradigm, as first introduced by Cand{\`e}s and Tao \cite{RefWorks:577} and Donoho \cite{RefWorks:70}, is based on using available degrees of freedom in a sensing mechanism to tune the measurement system so as to allow for efficient recovery of a particular type of signal or image of interest from as few measurements as possible.  A model assumption that allows for such signal recovery is \emph{sparsity}: the signal can be well-approximated by just a few elements of a given representation system.

Often, a near-optimal strategy is to choose the measurements completely at random, for example following a Gaussian distribution \cite{RefWorks:285}. Typically, however, some additional structure is imposed by the application at hand, and randomness can only be infused in the remaining degrees of freedom. For example, magnetic resonance imaging (MRI) is known to be well modeled by inner products with Fourier basis vectors. This structure cannot be changed, and the only aspect that can be decided at random is which Fourier basis vectors to select.

An important difference between completely random measurement systems and structured random measurement systems is in the aspect of universality. Gaussian measurement systems, among many other systems with minimal imposed structure, are oblivious to the basis in which the underlying signal is sparse, and achieve equal reconstruction quality for all different orthonormal basis representations. For structured measurement systems, this is, in general, no longer the case. When the measurements are uniformly subsampled from an orthonormal basis such as the Fourier basis, one requires, for example, that the measurement basis and the sparsity basis are \emph{incoherent}; the inner products between vectors from the two bases are small. If the two bases are not incoherent, more refined concepts of incoherence are required \cite{KW14, adcock2013breaking}. An important example is that of the Fourier measurement basis and a wavelet sparsity basis. Since both contain the constant vector, they are maximally coherent.

\subsection{Motivation}

All of the above mentioned works, however, exclusively cover the case of sparsity in \emph{orthonormal} basis representations.  On the other hand, there are a vast number of applications in which sparsity is expressed not in terms of a basis but in terms of a redundant, often highly overcomplete, dictionary.  Specifically, if $\f \in \C^{n}$ is the signal of interest to be recovered, then one expresses $\f = \D\x$, where $\D\in\C^{n\times N}$ is an overcomplete dictionary and $\x\in\C^N$ is a sparse (or nearly sparse) coefficient vector.  Redundancy is widespread in practice, either because no sparsifying orthonormal basis exists for the signal class of interest, or because the redundancy itself is useful and allows for a significantly larger, richer class of signals which can be sparsely represented in the resulting dictionary.  For example, it has been well documented that overcompleteness is the key to a drastic reduction in artifacts and recovery error in the denoising framework~\cite{RefWorks:342,RefWorks:344}.

In the compressive sensing framework, results using overcomplete dictionaries are motivated by the broad array of tight (and often Parseval) frames appearing in practical applications.  For example, if one assumes sparsity with respect to the Discrete Fourier Transform (DFT), this is implicitly assuming that the signal is well-represented by frequencies lying along the lattice of the DFT.  To allow for more flexibility in this rigid assumption, one instead may employ the oversampled DFT frame, containing frequencies on a much finer grid, or even over intervals of varying widths.  Gabor frames are used in imaging as well as radar and sonar applications, which are often highly redundant~\cite{RefWorks:224}.  Many of the most widely-used frames in imaging applications such as undecimated wavelet frames~\cite{RefWorks:345,RefWorks:342}, curvelets~\cite{RefWorks:273}, shearlets~\cite{labate2005sparse,easley2008sparse}, framelets~\cite{cai2008framelet}, and many others, are overcomplete with highly correlated columns.

\subsection{Related Work}

Due to the abundance of relevant applications, a number of works have studied compressive sensing for overcomplete frames. The first work on this topic aimed to recover the coefficient vector $\x$ directly, and thus required strong incoherence assumptions on the dictionary $\D$ \cite{RefWorks:583}. More recently, it was noted that if one instead aims to recover $\f$ rather than $\x$, recovery guarantees can be obtained under weaker assumptions. Namely, one only needs that the measurement matrix $\A$  respects the norms of signals which are sparse in the dictionary $\D$.  To quantify this, Cand\`es et al.~\cite{RefWorks:60} define the $\D$-restricted isometry property ($\D$-RIP in short, see Definition~\ref{DRIP-orig} below).  
For measurement matrices that have this property, a number of algorithms have been shown to guarantee recovery under certain assumptions.
Optimization approaches such as $\ell_1$-analysis~\cite{RefWorks:340,RefWorks:60,RefWorks:607,RefWorks:581,liu2012compressed,rauhut2013analysis,giryes2014sampling} and greedy approaches~\cite{Paper5,RefWorks:607,giryes2013greedy,peleg2013performance,giryes2014near} have been studied.  

This paper establishes the $\D$-RIP for structured random measurements formed by subsampling orthonormal bases, allowing for these types of recovery results to be utilized in more realistic settings. To date, most random matrix constructions known to yield $\D$-RIP matrices with high probability are random matrices with a certain concentration property. As shown in \cite{RefWorks:60}, such a property implies that for arbitrary dictionary $\D$, one obtains the $\D$-RIP with high probability. This can be interpreted as a dictionary version of the universality property discussed above.  It has now been shown that several classes of random matrices satisfy this property as well as subsampled structured matrices, after applying random column signs~\cite{Dasgupta, RefWorks:230}. The matrices in both of these cases are motivated by application scenarios, but typically in applications they appear without the randomized column signs.  In many cases one is not able to apply column signs in practice; for example in cases where 
the measurements are fixed such as in MRI, one has no choice but to use unsigned Fourier samples and cannot pre-process the data to incorporate column signs. Without these signs however, such measurement ensembles will not work for arbitrary dictionaries in general.  This is closely related to the underlying RIP matrix constructions not being universal. For example, it is clear that randomly subsampled Fourier measurements will fail for the oversampled Fourier dictionary (for reasons similar to the basis case). In this work, we address this issue, deriving recovery guarantees that take into account the dictionary. Similar 
to the basis case, our analysis will be coherence based. A similar approach has also been taken by Poon in \cite{P15} (completed after the first version of our paper), for an infinite dimensional version of the problem.

\subsection{Contribution}

Our main result shows that a wide class of orthogonal matrices having uniformly bounded entries can be subsampled to obtain a matrix that has the $\D$-RIP and hence yields recovery guarantees for sparse recovery in the setting of redundant dictionaries. 
As indicated above, our technical estimates below will imply such guarantees for various algorithms. As an example we focus on the method of $\ell_1$-analysis, for which the first $\D$-RIP based guarantees were available \cite{RefWorks:60}.  Our technical estimates will also provide more general guarantees for weighted $\ell_1$-analysis minimization (a weighted version of \eqref{eq:l1ana}, see \eqref{eq:l1anaW} in Section \ref{sec:main} for details) in case one has prior knowledge of the underlying sparsity support. 

Recall that the method of $\ell_1$-analysis consists of estimating a signal $\f$ from noisy measurements $\y = \A\f+\e$ by solving the convex minimization problem
\begin{equation}
 \f^{\sharp}= \argmin_{\tilde \f\in\C^n} \|\D^*\tilde\f\|_1 \text{\quad such that \quad} \|\A \tilde\f-\y\|_2\leq \varepsilon \tag{$P_1$}, \label{eq:l1ana}
\end{equation}
where $\varepsilon$ is the noise level, that is, $\|\e\|_2\leq \varepsilon$.  The $\ell_1$-analysis method (like alternative approaches) assumes  that for the signal $\f=\D\x$, not only is the underlying (synthesis) coefficient sequence $\x$ sparse (typically unknown and hard to obtain), but also the analysis coefficients $\D^*\f$ are nearly sparse, i.e., dominated by a few large entries. We refer the reader to Theorem~\ref{thm:l1anarec} below for the precise formulation of the resulting recovery guarantees (as derived in \cite{RefWorks:60}). The assumption has been observed empirically for many dictionaries used in practice such as the Gabor frame, undecimated wavelets, curvelets, etc. (see, e.g.,~\cite{RefWorks:60} for a detailed description of such frames) and is also key for a number of thresholding approaches to signal denoising \cite{RefWorks:273, shearlet07, easley2008sparse}.

A related and slightly stronger signal model, in which the analysis vector $\D^*\f$ is sparse or nearly sparse, has been considered independently from coefficient sparsity (e.g.,~\cite{RefWorks:581}), and is commonly called the co-sparsity model.

The results in this paper need a similar, but somewhat weaker assumption to hold for all signals corresponding to sparse synthesis coefficients $\x$.  Namely, one needs to control the \unrecoverable~as we now introduce. 
\begin{definition}\label{unrec-orig}
For a dictionary $\D\in\C^{n\times N}$ and a sparsity level $s$, we define the \textit{\unrecoverable} as
\begin{equation}\label{eq:unrec}
\unrec_{s, \D} =\unrec \defby \sup_{\enorm{\D\z}=1, \norm{\z}_{0}\leq s} \frac{\norm{\D^*\D\z}_{1}}{\sqrt{s}}.
\end{equation}
\end{definition}
In the following we will mainly consider dictionaries, which form Parseval frames, that is, $\D^*$ is an isometry.
Then the term \unrecoverable~is appropriate because this quantity can be viewed as the factor by which sparsity is preserved under the gram matrix map $\D^{*}\D$, compared to the case where $\D$ is an orthonormal basis and $\D^{*} \D = \Id_n$ and $\eta \equiv 1$ by Cauchy-Schwarz.  For a general family of such frames parameterized by the redundancy $N/n$, $\unrec$ will increase with the redundancy; 
families of dictionaries where such growth is relatively slow will be of interest.  For example, new results on the construction of unit-norm tight frames give a constructive method to generate \textit{spectral tetris} frames \cite{CFMWZ11}, whose gram matrices have at most $2\lceil N/n\rceil+6$ non-zeros per row, guaranteeing that $\unrec$ is proportional only to the redundancy factor $N/n$.  In fact, one can show that these are sparsest frames possible \cite{CHKK11}, suggesting that families of tight frames with \unrecoverable~depending linearly on the redundancy factor should be essentially optimal for $\ell_1$-analysis reconstruction.   We pose as a problem for subsequent work to obtain such estimates for dictionaries of practical interest, discussing a few examples in Section~\ref{sec:main}.  First, to illustrate that our results go strictly beyond existing theory, we show that harmonic frames (frames constructed by removing the high frequency rows of the DFT -- see \eqref{harmonic}) with small redundancy indeed have a bounded \unrecoverable. To our knowledge, this important case is not covered by existing theories. In addition, we also bound the \unrecoverable~for redundant Haar wavelet frames.  We expect, however, that it will be difficult to efficiently compute the \unrecoverable~of an arbitrary dictionary.
For bounded \unrecoverable s, we prove the following theorem.  

\begin{theorem}
 \label{thm:intro}
Fix a sparsity level $s < N$.  Let $\D\in\C^{n\times N}$ be a Parseval frame -- i.~e., $\D\D^*$ is the identity -- with columns $\{\d_1, \ldots, \d_N\}$,  and let $\B=\{\b_1,\dots, \b_n\}$ be an orthonormal basis of $\C^n$ which is {\em incoherent} to $\D$ in the sense that 
\begin{equation}\label{eq:incoh}
\sup_{i \in [n]} \sup_{ j\in [N]} | \langle{ \b_i, \d_j \rangle} | \leq K n^{-1/2}
\end{equation}
for some constant $K \geq 1$.

Consider the (unweighted) \unrecoverable~ $\unrec = \unrec_{s,\D}$ as in Definition~\ref{unrec-orig}.  Construct $\widetilde{\B}\in\C^{m\times n}$ by sampling row vectors from $\B$ i.i.d. uniformly at random. 
Provided
\begin{equation}\label{mboundloc}
m \geq C s K^2 \unrec^2\log^3(s\unrec^2) \log(N), \quad 
\end{equation}
then with probability $1-N^{-\log^3 (2s)}$, $\sqrt{\frac{n}{m}}\widetilde \B$ 
exhibits uniform recovery guarantees for $\ell_1$-analysis. That is, for every signal $\f$, the solution $\f^{\sharp}$ of the minimization problem~\eqref{eq:l1ana} with $\y=\sqrt{\frac{n}{m}}\widetilde\B\f+\e$ for noise $\e$ with $\|\e\|_2\leq \varepsilon$ satisfies
\begin{equation}\label{eq:intrec}
 \|\f^{\sharp}-\f\|_2 \leq C_1 \varepsilon + C_2 \frac{\|\D^*\f - (\D^*\f)_s \|_1}{\sqrt{s}}.
\end{equation}
Here, $C, C_1$, and $C_2$ are absolute constants independent of the dimensions and the signal $\f$.
\end{theorem}
Above and throughout, $(\u)_s$ denotes the best $s$-sparse approximation to a signal $\u$, that is, $(\u)_{s} \defby \emph{arg}\min_{\norm{\z}_{0} \leq s}\enorm{\u-\z}$.  

\noindent{\bf Remarks.}\\  
\noindent{\bf 1.}  Our stability result \eqref{eq:intrec} for $\ell_1$ analysis guarantees recovery for the particular $\f = \D\z$ only up to the scaled $\ell_1$ norm of the tail of $\D^*\D\z$.  In fact, the quantity $\mathcal{E}^{*}_{\f} = \frac{\|\D^*\f - (\D^*\f)_s\|_1}{\sqrt{s}}$, referred to as the \emph{unrecoverable energy} of the signal $\f$ in $\D$ \cite{RefWorks:60}, is closely related to the \unrecoverable~$\eta_{s}$:
$$\eta_{s} = \sup_{\f = \D \z: \hspace{.5mm} \| \f \|_2 = 1, \| \z \|_0 \leq s} \frac{\|\D^*\f \|_1}{\sqrt{s}} \leq  \sup_{\f = \D \z: \hspace{.5mm} \| \f \|_2 = 1, \| \z \|_0 \leq s} \mathcal{E}^{*}_{f} + 1$$ 

\noindent{\bf 2.}  As mentioned, the proof of this theorem will proceed via the $\D$-RIP, so our analysis yields similar guarantees for other recovery methods as well.  

\noindent{\bf 3.}  It is interesting to remark on the role of incoherence in this setting.  While prior work in compressed sensing has required the measurement matrix $\widetilde\B$ itself be incoherent, we now ask instead for incoherence between the basis from which measurements are selected and the dictionary $\D$, rather than incoherence within $\D$.  Of course, note that the coherence of the dictionary $\D$ itself impacts the localization factor $\unrec_{s, \D}$.  We will see later in Theorem~\ref{thm:main} that the incoherence requirement between the basis and dictionary can be weakened even further by using a weighted approach.

\noindent{\bf 4.}  The assumption that $\D$ is a Parseval frame is not necessary but made for convenience of presentation.  Several results have been shown for frames which are not tight, e.g.,  \cite{liu2012compressed, rauhut2013analysis, giryes2014sampling}.  Indeed, any frame $\D \in\C^{n\times N}$ with linearly independent rows is such that
$$
\tilde{\D} := (\D \D^{*})^{-1/2} \D
$$
is a tight frame.  As observed in the recent paper \cite{Foucart15}, in this case one can use for sampling the adapted matrix $\tilde{\A} = (\tilde{\B} (\D \D^{*})^{-1/2}),$ as measuring the signal $\f = \D \z$  through $\y = \tilde{\A} \f +  \e$ is equivalent to measuring the signal $\tilde{\f} = \tilde{\D} \z$ through $\y = \tilde{\B} \tilde{\f} +  \e$.   Working through this extension poses an interesting direction for future work.

\subsection{Organization}
The rest of the paper is organized as follows. We introduce some notation and review some technical background in Section~\ref{sec:notate} before presenting a technical version of our main result in Section~\ref{sec:main}, which demonstrates that bases which are incoherent with a dictionary $\D$ can be randomly subsampled to obtain $\D$-RIP matrices.  In that section, we also discuss some relevant applications and the implications of our results, including the regime where the sampling basis is coherent to the dictionary. The proof of our main theorem is presented in the final section.

\section{Notation and technical background}\label{sec:notate}

Throughout the paper, we write $C$, $C'$, $C_1$, $C_2,\ldots$ to denote absolute constants; their values can change between different occurrences. We write $\Id_n$ to denote the $n\times n$ identity matrix and we denote by $[n]$ the set $\{1, 2, \ldots, n\}$.   For a vector $\u$, we denote the $j$th index by $\u[j]$ or $\u(j)$ or $\u_j,$ depending on which notation is the most clear in any given context; its restriction to only those entries indexed by a subset $\Lambda$ is donated by $u_\Lambda$.   These notations should not be confused with the notation $(\u)_s$, which we use to denote the best $s$-sparse approximation to $\u$, that is, $(\u)_{s} \defby \text{arg}\inf_{\norm{\z}_{0} \leq s}\enorm{\u-\z}$. Similarly for a matrix $\A$, $\A_j$ is the $j$th column and $\A_\Lambda$ is the submatrices consisting of the columns with indices in $\Lambda$.

\subsection{The unweighted case with incoherence} 
Recall that a dictionary $\D\in\C^{n\times N}$ is a {\em Parseval frame} if $\D\D^*=\Id_n$ and that a vector $\x$ is {\em $s$-sparse} if $\|\x\|_0 := |\supp(\x)| \leq s$. Then the restricted isometry property with respect to the dictionary $\D$ ($\D$-RIP) is defined as follows.
 \begin{definition}[\cite{RefWorks:60}]\label{DRIP-orig}
Fix a dictionary $\D\in\C^{n\times N}$ and matrix $\A\in\C^{m\times n}$.  The matrix $\A$ satisfies the $\D$-RIP with parameters $\delta$ and $s$ if
\begin{equation}\label{DRIPbound-orig}
(1-\delta)\enormsq{\D\x} \leq \enormsq{\A\D\x} \leq (1+\delta)\enormsq{\D\x}
\end{equation}
for all $s$-sparse vectors $\x\in\C^N$.
\end{definition}

Note that when $\D$ is the identity, this definition reduces to the standard definition of the restricted isometry property~\cite{RefWorks:48}. Under such an assumption on the measurement matrix, the following results bound the reconstruction error for the $\ell_1$-analysis method.

\begin{theorem}[\cite{RefWorks:60}]\label{thm:l1anarec}
Let $\D$ be a tight frame, $\varepsilon>0$, and consider a matrix $\A$ that has the $\D$-RIP with parameters $2s$ and $\delta<0.08$. 
Then for every signal $\f\in\C^n$, the reconstruction $\f^{\sharp}$ obtained from noisy measurements $\y=\A\f+\e$, $\|\e\|_2\leq \varepsilon$, via the $\ell_1$-analysis problem \eqref{eq:l1ana} above satisfies

\begin{equation}\label{L1bound}
   \enorm{\f-\f^{\sharp}} \leq \varepsilon + \frac{\norm{\D^*\f - (\D^*\f)_s}_{1}}{\sqrt{s}}.
\end{equation}
\end{theorem}

 Thus as indicated in the introduction, one obtains good recovery for signals $\f$ whose analysis coefficients $\D^*\f$ have a suitable decay.

Both the RIP and the $\D$-RIP are closely related to the Johnson-Lindenstrauss lemma~\cite{RefWorks:327,RefWorks:326,RefWorks:86}.  Recall the classical variant of the lemma states that for any $\varepsilon\in (0, 0.5)$ and points $\x_1, \ldots, \x_d \in \mathbb{R}^n$, that there exists a Lipschitz function $f : \mathbb{R}^n \rightarrow \mathbb{R}^m$ for some $m = O(\varepsilon^{-2}\log d)$ such that
\begin{equation}\label{JL}
(1-\varepsilon) \|\x_i - \x_j\|_2^2 \leq \|f(\x_i) - f(\x_j)\|_2^2 \leq (1+\varepsilon)\|\x_i-\x_j\|_2^2
\end{equation}  
for all $i, j \in \{1, \ldots, d\}$.  Recent improvements have been made to this statement which in particular show that the map $f$ can be taken as a random linear mapping which satisfies \eqref{JL} with high probability (see, e.g., \cite{RefWorks:90} for an account of such improvements).
Indeed, any matrix $\A\in\C^{m\times n}$ which for a fixed\footnote{By ``fixed'' we mean to emphasize that this probability bound must occur for a \textit{single} vector $\z$ rather than for all vectors $\z$ like one typically sees in the restricted isometry property.} vector $\z\in\C^n$ satisfies
$$
\P\left((1-\delta)\enormsq{\z} \leq \enormsq{\A\z} \leq (1+\delta)\enormsq{\z}\right) \leq Ce^{-cm}
$$
will satisfy the $\D$-RIP with high probability as long as $m$ is at least on the order of $s\log(n/s)$~\cite{RefWorks:60}.  From this, any matrix satisfying the Johnson-Lindenstrauss lemma will also satisfy the $\D$-RIP (see \cite{RefWorks:86} for the proof of the RIP for such matrices, which directly carries over to the $\D$-RIP).  Random matrices known to have this property include matrices with independent subgaussian entries (such as Gaussian or Bernoulli matrices), see for example \cite{Dasgupta}.  Moreover, it is shown in \cite{RefWorks:230} that any matrix that satisfies the classical RIP will satisfy the Johnson-Lindenstrauss lemma and thus the $\D$-RIP with high probability after randomizing the signs of the columns.  The latter construction allows for structured random matrices with fast multiplication properties such as randomly subsampled Fourier matrices (in combination with the results from \cite{RefWorks:285}) and matrices representing subsampled random convolutions (in 
combination with the results from \cite{rauhut2012restricted, KMR12}); in both cases, however, again with randomized column signs.
 While this gives an abundance of such matrices, as mentioned above, it is not always practical or possible to apply random column signs in the sampling procedure.  

An important general set-up for structured random sensing matrices known to satisfy the regular RIP is the framework of \emph{bounded orthonormal systems}, which includes as a special case the subsampled discrete Fourier transform measurements (without column signs randomized).  Such measurements are the only natural measurements possible in many physical systems where compressive sensing is of interest, such as in MRI, radar, and astronomy \cite{RefWorks:82,RefWorks:79,RefWorks:289,bobin2008compressed}, as well as in applications to polynomial interpolation \cite{rw11, burq2012weighted, rauhut2011sparse} and uncertainty quantification \cite{hampton2015compressive}. In the following, we recall this set-up (in the discrete setting), see \cite[Sec.4]{ra09-1} for a detailed account of bounded orthonormal systems and examples.

\begin{definition}[Bounded orthonormal system]
\label{BONS}
Consider a probability measure $\nu$ on the discrete set $[n]$ and a system $\{ \r_j \in \C^n$, $j \in [n] \}$ that is orthonormal with respect to $\nu$ in the sense that
$$
\sum_{i} \r_k(i) \r_j(i) \nu_i = \delta_{j,k}
$$
where $\delta_{j,k}=\begin{cases}
                             1 & \text{ if } j = k\\
                             0 & \text{ else.}
                            \end{cases}$ 
 is the Kronecker delta function.  Suppose further that the system is uniformly bounded: there exists a constant $K \geq 1$ such that 
\begin{equation}
\label{local}
\sup_{i \in [n]} \sup_{j \in [n]} | \r_j(i) | \leq K.
\end{equation} 
Then the matrix $\A \in \C^{n \times n}$ whose rows are $\r_j$ is called a \emph{bounded orthonormal system} matrix.
\end{definition}

Drawing $m$ indices $i_1, i_2, \dots, i_m$ independently from the orthogonalization measure $\nu$, the sampling matrix $\widetilde{\A} \in \C^{m \times n}$ whose rows are indexed by the (re-normalized) sampled vectors $\sqrt{\frac{1}{m}} \r_j(i_k) \in \C^n$ will have the restricted isometry property with high probability (precisely, with probability exceeding $1-n^{-C \log^3(s)}$) provided the number of measurements satisfies 
\begin{equation}
\label{BOS:RIP:cond}
m \geq CK^2 s \log^2(s) \log(n).
\end{equation}
This result was first shown in the case where $\nu$ is the uniform measure by Rudelson and Vershynin~\cite{RefWorks:285}, for a slightly worse dependence of $m$ on the order of $s\log^2 s \log(s\log n) \log(n)$.  These results were subsequently extended to the general bounded orthonormal system set-up by Rauhut \cite{ra09-1}, and the dependence of $m$ was slightly improved to $s\log^3 s\log n$ in \cite{rauhut2013interpolation}. 

An important special case where these results can be applied is that of {\em incoherence} between the measurement and sampling bases. Here the {\em coherence} between two sets of vectors $\Phi=\{\phi_i\}$ and $\Psi=\{\psi_j\}$  is given by $\mu=\sup_{i,j} |\langle \phi_i, \psi_j\rangle|$. Two orthonormal bases $\Phi$ and $\Psi$ of $\C^n$ are called incoherent if $\mu \leq K n^{-1/2}$. In this case, the renormalized system $\tilde\Phi=\{\sqrt{n} \phi_i \Psi^*\}$ is an orthonormal system with respect to the uniform measure, which is bounded by $K$. Then the above results imply that signals which are sparse in basis $\Psi$ can be reconstructed from inner products with a uniformly subsampled subset of basis $\Phi$. These incoherence-based guarantees are a standard criterion to ensure signal recovery, such results had first been observed in \cite{RefWorks:379}.

\subsection{Generalization to a weighted setup and to local coherence} Recently, the criterion of incoherence has been generalized to the case where only \emph{most} of the inner products between sparsity basis vectors and measurement basis vectors are bounded. If some vectors in the measurement basis yield larger inner products, one can adjust the sampling measure and work with a preconditioned measurement matrix \cite{KW14}, and if some vectors in the sparsity basis yield larger inner products, one can incorporate weights into the recovery schemes and work with a weighted sparsity model \cite{rauhut2013interpolation}. Our results can accommodate both these modifications. In the remainder of this section, we will recall some background on these frameworks from \cite{KW14, rauhut2013interpolation} and formulate our definitions for the more general setup.

Consider a set of positive weights $\omega = (\omega_j)_{j \in [N]}$.  Associate to these weights the norms
\begin{equation}
\label{weighted_norm}
\| \x \|_{\omega,p} := \left( \sum_j | x_j |^p \omega_j^{2-p} \right)^{1/p}, \quad 0 < p \leq 2,
\end{equation}
along with the weighted ``$\ell_0$-norm", or weighted sparsity: $\| \x \|_{\omega,0} := \sum_{j: | x_j | > 0} \omega_j^2$, which equivalently is defined as the limit $\| \x \|_{\omega,0} = \lim_{p \rightarrow 0} \| \x \|_{\omega,p}^p$. We say that a vector $\x$ is {\bf weighted $s$-sparse} if $\| \x \|_{\omega,0} := \sum_{j: | x_j | > 0} \omega_j^2 \leq s$. In line with this definition, the weighted size of a finite set $\Lambda \subset \NN$, is given by $\omega(\Lambda):= \sum_{j\in\Lambda} \omega_j^2$; thus a vector is weighted $s$-sparse if its support has weighted size at most $s$.  When $\omega_j \geq 1$ for all $j$, the weighted sparsity of a vector is at least as large as its unweighted sparsity, so that the class of weighted $s$-sparse vectors is a strict subset of the $s$-sparse vectors. We make this assumption in the remainder.
Note in particular the special cases $\| \x \|_{\omega,1} = \sum_j | x_j | \omega_j$ and $\| \x \|_{\omega,2} = \| \x \|_{2} = \sqrt{\sum_j | x_j |^2};$ by Cauchy-Schwarz, it follows that $ \| \x \|_{\omega,1} \leq \sqrt{s} \| \x \|_2$ if $\x$ is weighted $s$-sparse.   Indeed, we can extend the notions of \unrecoverable~ \eqref{eq:unrec} and $\D$-RIP \eqref{DRIP-orig} to the weighted sparsity setting.  It should be clear from context which definition we refer to in the remainder.

\begin{definition}\label{eq:Wunrec}
For a dictionary $\D\in\C^{n\times N}$, weights $\omega_1, \omega_2, \dots, \omega_N \geq 1$, and a sparsity level $s$, we define the (weighted) \textit{\unrecoverable} as
\begin{equation}
\label{unrec-full}
\unrec_{\omega, s, \D} =\unrec \defby \sup_{\enorm{\D\z}=1, \norm{\z}_{\omega,0}\leq s} \frac{\norm{\D^*\D\z}_{\omega,1}}{\sqrt{s}},
\end{equation}
\end{definition}

\begin{definition}\label{DRIP}
Fix a dictionary $\D\in\C^{n\times N}$, weights $\omega_1, \omega_2, \dots, \omega_N \geq 1$, and matrix $\A\in\C^{m\times n}$.  The matrix $\A$ satisfies the $\D$-$\omega$RIP with parameters $\delta$ and $s$ if
\begin{equation}\label{DRIPbound}
(1-\delta)\enormsq{\D\x} \leq \enormsq{\A\D\x} \leq (1+\delta)\enormsq{\D\x}
\end{equation}
for all {\bf weighted $s$-sparse} vectors $\x\in\C^N$.
\end{definition}
When $\D$ is an orthonormal matrix, this definition reduces to the definition $\omega$RIP from \cite{rauhut2013interpolation}.   
More generally, weights allow the flexibility to incorporate prior information about the support set and allow for weaker assumptions on the dictionary $\D$.    
In particular, a larger weight assigned to a dictionary element will allow for this element to have larger inner products with the measurement basis vectors. In this regard, a basis version of our result is the following variant of a theorem from \cite{rauhut2013interpolation}:

\begin{proposition}[from \cite{rauhut2013interpolation}]
\label{weights:RW}
Fix a probability measure $\nu$ on $[n]$, sparsity level $s < n$, and constant $0 < \delta < 1$.  Let $\D\in\C^{n \times n}$ be an orthonormal matrix.  Let $\A$ be an orthonormal system with respect to $\nu$ as in Definition~\ref{BONS} with rows denoted by $\r_i$, and consider weights $\omega_1, \omega_2, \dots, \omega_n \geq 1$ such that $\max_i | \langle{ \r_i, \d_j \rangle} | \leq \omega_j$.  Construct an $m\times n$ submatrix $\widetilde{\A}$ of $\A$ by sampling rows of $\A$ according to the measure $\nu$ and renormalizing each row by $\sqrt{\frac{1}{m}}$.  Provided
\begin{equation}\label{mbound}
m \geq C\delta^{-2} s \log^3(s) \log(n),
\end{equation}
then with probability $1-n^{c\log^3 s}$, the submatrix $\widetilde{\A}$ satisfies the $\omega$RIP with parameters $s$ and $\delta$.
\end{proposition}

An adjusted sampling density for the case when certain vectors in the measurement basis yield larger inner products is obtained via the {\bf local coherence} as defined in the following.
\begin{definition}[Local coherence, \cite{KW14, rw11}]
The \emph{local coherence} of a set $\Phi=\{\varphi_i\}_{i=1 }^k \subset \C^n$ with respect to another set $\Psi=\{\psi_j\}_{j=1 }^\ell \subseteq \C^n$ is the function $\mu^{\text{loc}}(\Phi, \Psi) \in \R^k$ defined coordinate-wise by 

{
\begin{equation*} \mu^{\text{loc}}_i(\Phi, \Psi) =  \mu^{{\text loc}}_i = \sup\limits_{1\leq j\leq \ell } |\langle \varphi_i, \psi_j\rangle|, \quad \quad  i = 1,2,\dots, k.
\end{equation*}
}
\end{definition}
When $\Phi$ and $\Psi$ are orthonormal bases, and a subset of $\Phi$ is used for sampling while $\Psi$ is the basis in which sparsity is assumed, the main point is that renormalizing the different vectors $\varphi_i$in the sampling basis by respective factors $\frac{1}{\mu^{{\text loc}}_i}$ does not affect orthogonality of the rows of $\Phi^{*} \Psi$. In this way, the larger inner products can be reduced in size. To compensate for this renormalization and retain the orthonormal system property in the sense of Definition~\ref{BONS}, one can then adjust the sampling measure. The resulting preconditioned measurement system with variable sampling density will then yield better bounds $K$ in Equation~\eqref{BOS:RIP:cond}. This yields the following estimate for the restricted isometry property.

\begin{proposition}[\cite{KW14, rw11}]\label{thm:wBOS}
 Let $\Phi=\{\varphi_j\}_{j=1}^n$ and  $\Psi =\{\psi_k\}_{k=1}^n$ be orthonormal bases of $\C^n$. Assume the local coherence of $\Phi$ with respect to $\Psi$  is pointwise bounded by the function $\kappa/\sqrt{n} \in \C^n$, that is  $ \sup\limits_{1\leq k\leq n} |\langle \varphi_j, \psi_k\rangle| \leq \kappa_j$. 
Suppose
 \begin{equation}\label{wBOS:RIP:cond}
m \geq C\delta^{-2} \|\kappa \|_2^2 s \log^3(s) \log(n),
\end{equation}
and choose $m$ (possibly not distinct) indices $j\in [n]$  i.i.d. from the probability measure $\nu$ on $[n]$ given by
\begin{equation*}
 \nu(j) = \frac{\kappa^2_j}{\|\kappa \|_2^2 }.
\end{equation*}
 Call this collection of selected indices $\Omega$ (which may possibly be a multiset).  Consider the matrix $\A \in \C^{m \times n}$ with entries
\begin{equation}\label{def:wPhi:matrix}
A_{j,k} = \langle \varphi_j, \psi_k\rangle, \quad j \in \Omega, k \in [n],
\end{equation}
and consider the diagonal matrix $\W = \operatorname{diag}(\w) \in \C^{n \times n}$ with $w_{j} = \| \kappa \|_2 / \kappa_j$.
Then with probability at least 
$1-n^{-c \log^3(s)},$ 
the preconditioned matrix $\frac{1}{\sqrt{m}} \W \A$ has the restricted isometry property with parameters $\delta$ and $s$.
\end{proposition}

\noindent{\bf Remark.}  
In case $\Phi$ and $\Psi$ are incoherent, or if $\kappa_j = K n^{-1/2}$ uniformly for all $j$, then local coherence sampling as above reduces to the previous results for incoherent systems: in this case, the associated probability measure $\nu$ is uniform, and the preconditioned matrix reduces to $\frac{1}{\sqrt{m}} \W \A = \sqrt{ \frac{n}{m} } \A$. 

Our main result on $\D$-RIP for redundant systems and structured measurements (Theorem \ref{thm:main}) implies a strategy for extending local coherence sampling theory to dictionaries, of which Theorem \ref{thm:intro} is a special case.  Indeed, if $\Psi = \{ \psi_j \}$, more generally, is a Parseval frame and $\Phi = \{\varphi_i \}$ is an orthonormal matrix, then renormalizing the different vectors $\varphi_i$ in the sampling basis by respective factors $\frac{1}{\mu^{{\text loc}}_i}$ still does not affect orthogonality of the rows of $\Phi^{*} \Psi$.   We have the following corollary of our main result; we will establish a more general result, Theorem \ref{thm:main} in the next section, from which both Theorem \ref{thm:intro} and Corollary \ref{thm:variable} follow.

\begin{corollary}\label{thm:variable}
Fix a sparsity level $s < N$, and constant $0 < \delta < 1$.  Let $\D \in\C^{n\times N}$ be a Parseval frame with columns $\{\d_1, \ldots, \d_N\}$,  let $\B$ with rows  $\{\b_1,\dots, \b_n\}$ be an orthonormal basis of $\C^n$, and assume the local coherence of $\D$ with respect to $\B$  is pointwise bounded by the function $\kappa \in \C^n$, that is  $ \sup\limits_{1\leq j \leq N} |\langle \d_j, \b_k\rangle| \leq \kappa_k$. 
 
Consider the probability measure $\nu$ on $[n]$ given by $\nu(k) = \frac{\kappa^2_k}{\|\kappa \|_2^2 }$ 
along with the diagonal matrix $\W = \operatorname{diag}(\w) \in \C^{n \times n}$ with $w_{k} = \| \kappa \|_2 / \kappa_k$.  Note that the renormalized system $\Phi = \frac{1}{\sqrt{m}} \W \B$ is an orthonormal system with respect to $\nu$, bounded by $\| \kappa \|^2$.  
Construct $\widetilde{\B}\in\C^{m\times n}$ by sampling vectors from $\B$ i.i.d. from the measure $\nu$, and consider the \unrecoverable~ $\unrec = \unrec_{\D,s}$ of the frame $\D$ as in Definition~\ref{unrec-orig}, 

As long as 
$$
m \geq C\delta^{-2} \unrec^2 \|\kappa \|_2^2 s \log^3(s \unrec^2) \log(N)
$$
then with probability $1-N^{-\log^3(2s)}$, the following holds for every signal $\f$: the solution $\f^{\sharp}$ of the weighted $\ell_1$-analysis problem
\begin{equation}
 \f^{\sharp}= \argmin_{\tilde \f\in\C^n} \|\D^*\tilde\f\|_1 \text{\quad such that \quad} \| \frac{1}{\sqrt{m}} \W (\widetilde{\B} \tilde\f-\y) \|_2\leq \varepsilon \tag{$P_{1,w}$}, \label{eq:l1anaW}
\end{equation}
with $\y=\B \f+\e$ for noise $\e$ with weighted error $\| \frac{1}{\sqrt{m}} \W \e \|_2 \leq \varepsilon$ satisfies
\begin{equation}
 \|\f^{\sharp}-\f\|_2 \leq C_1 \varepsilon + C_2 \frac{\|\D^*\f - (\D^*\f)_s\|_1}{\sqrt{s}}.
\end{equation}
Here $C, C_1$, and $C_2$ are absolute constants independent of the dimensions and the signal $\f$.
\end{corollary}

The resulting bound involves a weighted noise model, which can, in the worst case, introduce an additional factor of 
$\sqrt{\frac{n}{m}}\max_i \sqrt{\nu(i)}$.   We believe however that this noise model is just an artifact of the proof technique, and that the stability results in Corollary \ref{thm:variable} should hold for the standard noise model using the standard $\ell_1$-analysis problem \eqref{eq:l1ana}.  In the important case where $\D$ is a wavelet frame and $\B$ is the orthonormal discrete Fourier matrix, we believe that \emph{total variation minimization}, like $\ell_1$-analysis, should give stable and robust error guarantees with the standard measurement noise model.  Indeed, such results were recently obtained for variable density sampling in the case of orthonormal wavelet sparsity basis \cite{poon14}, improving on previous bounds for total variation minimization \cite{KW14, needell2013stable, needell2013near, adcock2013breaking}.   Generalizing such results to the dictionary setting is indeed an interesting direction for future research.

	\section{Our main result and its applications}\label{sec:main}
As mentioned in the previous section, the case of incoherence between the sampling basis and the sparsity dictionary, as in Theorem~\ref{thm:intro}, is a special case of a bounded orthonormal system. The following more technical formulation of our main result in the framework of such systems covers both weighted sparsity models and local coherence and as we will see, implies Theorem~\ref{thm:intro}.  The proof of Theorem~\ref{thm:main} is presented in Section~\ref{sec:proof}.

\begin{theorem}\label{thm:main}
Fix a probability measure $\nu$ on $[N]$, sparsity level $s < N$, and constant $0 < \delta < 1$.  Let $\D\in\C^{n\times N}$ be a Parseval frame.  Let $\A$ be an orthonormal systems matrix with respect to $\nu$ with rows $\r_i$ as in Definition~\ref{BONS}, and consider weights $\omega_1, \omega_2, \dots, \omega_N \geq 1$ such that 
\begin{equation}\label{eq:wts}
\max_i | \langle{ \r_i, \d_j \rangle} | \leq \omega_j.
\end{equation}
 Define the unrecoverable energy $\unrec = \unrec_{\D,s}$ as in Definition~\ref{unrec-full}.  Construct an $m\times n$ submatrix $\widetilde{\A}$ of $\A$ by sampling rows of $\A$ according to the measure $\nu$.  Then as long as 
\begin{eqnarray}\label{mbound3}
m &\geq& C\delta^{-2} s \unrec^2\log^3(s \unrec^2) \log(N), \quad \text{and} \nonumber \\
m &\geq& C\delta^{-2} s \unrec^2 \log(1/\gamma)
\end{eqnarray}
then with probability $1-\gamma$, the normalized submatrix $\sqrt{\tfrac{1}{m}} \widetilde{\A}$ satisfies the $\D$-$\omega$RIP with parameters $s$ and $\delta$.
\end{theorem}
 \subsection*{Proof of Theorem~\ref{thm:intro} assuming Theorem~\ref{thm:main}:} 
Let $\B$, $\D$ and $s$ be given as in Theorem~\ref{thm:intro}.  We will apply Theorem~\ref{thm:main} with $\omega_j = 1$ for all $j$, with $\nu$ the uniform measure 
on $[n]$, sparsity level $2s$, $\gamma = N^{-\log^3 (2s)}$, $\delta = 0.08$, and matrix $\A=\sqrt{n}\B$.  We first note that since $\B$ is an orthonormal basis, that $\sqrt{n}\B$ is an orthonormal systems matrix with respect to the uniform measure $\nu$ as in Definition~\ref{BONS}.  In addition, \eqref{eq:incoh} implies \eqref{eq:wts} for matrix $\A=\sqrt{n}\B$ with $K=\omega_j = 1$.  Furthermore, setting $\gamma = N^{-\log^3(2s)}$, \eqref{mboundloc} implies both inequalities of~\eqref{mbound3} (adjusting constants appropriately).  Thus the assumptions of Theorem~\ref{thm:main} are in force.  Theorem~\ref{thm:main} then guarantees that with probability $1-\gamma  = 1 - N^{-\log^3 (2s)}$, the uniformly subsampled matrix $\sqrt{\frac{1}{m}}(\sqrt{n}\widetilde{\B})$ satisfies the $\D$-$\omega$RIP with parameters $2s$ and $\delta$ and weights $\omega_j = 1$.  A simple calculation and the definition of the $\D$-$\omega$RIP (Definition~\ref{DRIP}) shows that this implies the $\D$-RIP with parameters $2s$ and $\delta$.  By Theorem~\ref{thm:l1anarec}, \eqref{eq:intrec} holds and this completes the proof.
\hfill $\square$ 
\subsection*{Proof of Corollary \ref{thm:variable} assuming Theorem \ref{thm:main}}
Corollary \ref{thm:variable} is implied by Theorem~\ref{thm:main} because the preconditioned matrix $\Phi = \frac{1}{\sqrt{m}} \W \B$ formed by normalizing 
the $i$th row of $\B$ by  $ \| \kappa \|_2 / \kappa_j$ constitutes an orthonormal system with respect to the probability measure $\nu(i)$ and uniform weights $\omega_j = \| \kappa \|_2$.
Then the preconditioned sampled matrix $\widetilde{\Phi}$ satisfies the $\D$-RIP for sparsity $2s$ and parameter $\delta$ according to Theorem \ref{thm:main} with probability $1-\gamma  = 1 - N^{-\log^3 (2s)}$.  Applying Theorem~\ref{thm:l1anarec} to the preconditioned sampling matrix $\widetilde{\Phi}$ and Parseval frame $\D$ produces the results in Corollary \ref{thm:variable}.
\hfill $\square$ 

\subsection{Example: Harmonic frame with $L$ more vectors than dimensions}

 It remains to find examples of a dictionary with bounded \unrecoverable~ and an associated measurement system for which incoherence condition \eqref{eq:incoh} holds. Our main example is that of sampling a signal that is sparse in an oversampled Fourier system, a so-called {\em harmonic frame} \cite{vale2004tight}; the measurement system is the standard basis.  Indeed, one can see by direct calculation that the standard basis is incoherent in the sense of \eqref{eq:incoh} to any set of Fourier vectors, even of non-integer frequencies.
We will now show that if the number of frame vectors exceeds the dimension only by a constant, such a dictionary will also have bounded \unrecoverable. This setup is a simple example, but our results apply, and it is not covered by previous theory.  

More precisely, we fix $L\in\NN$ and consider $N=n+L$ vectors in dimension $n$. We assume that $L$ is such that $Ls\leq \tfrac{N}{4}$. 
Then the harmonic frame is defined via its frame matrix $\D=(d_{jk})$, which results from the $N\times N$ discrete Fourier transform matrix $\F=(f_{jk})$ by deleting the last $L$ rows. That is, we have 
\begin{equation}\label{harmonic}
d_{jk} = \tfrac{1}{\sqrt{n+L}}\exp(\tfrac{2\pi i j k}{n+L})
\end{equation} 
for $j=1\ldots n$ and $k=1\ldots N$. The corresponding Gram matrix satisfies $(\D^*\D)_{ii}=\tfrac{n}{n+L}$ and, for $j\neq k$, by orthogonality of $\F$,
\[
        |(\D^*\D)_{jk}|=  \left|(\F^*\F)_{jk}  -   \sum_{\ell=n+1}^N \bar f_{\ell j} f_{\ell k}   \right|  = \left|  \sum_{\ell=n+1}^N \tfrac{1}{n+L} \exp(\tfrac{-2\pi i \ell j}{n+L}) \exp(\tfrac{2\pi i \ell k}{n+L})   \right|  \leq   \tfrac{L}{n+L}                                         
\]
As a consequence, we have for $\z$ $s$-sparse and $j\notin \supp \z$
\[
 |(\D^*\D \z)[j]| = \left|\sum_{k\in\supp\z} (\D^*\D)_{jk} \z[k] \right|\leq \tfrac{L}{n+L} \|\z\|_1
\]
where we write $\z[k]$ to denote the $k$th index of the vector $\z.$
Similarly, for $j\in \supp \z$, we obtain
\[
 |(\D^*\D \z)[j] - \z[j]| = \left|\sum_{k\in\supp\z, k\neq j} (\D^*\D)_{jk} \z[k] \right|\leq \tfrac{L}{n+L} \|\z\|_1.
\]
So we obtain, using that $\D^*$ is an isometry,
\begin{align*}
 \|\D\z\|^2_2 & =\|\D^*\D\z\|^2_2\geq \sum _{k\in \supp\z}\left( (\D^*\D \z)[k] \right)^2 
\geq \sum _{k\in \supp\z}(|\z[k]|-|(\D^*\D \z)[k] - \z[k]|)^2 \\
&\geq \sum _{k\in \supp\z}|\z[k] |^2 - 2 \tfrac{L}{n+L} \|\z\|_1 |\z[k] |\\ 
&= \|\z\|_2^2-\tfrac{2L}{n+L} \|\z\|_1^2 \geq (1-\tfrac{2Ls}{N})\|\z\|_2^2\geq \tfrac{1}{2} \|\z\|_2^2.
\end{align*}
That is, for $\z$ with $\|\D\z\|_2=1$, one has $\|\z\|_2\leq \sqrt{2}$.   Consequently, 
\begin{align*}
 \unrec&=\sup_{\enorm{\D\z}=1, \norm{\z}_{0}\leq s} \frac{\norm{\D^*\D\z}_{1}}{\sqrt{s}} = \sup_{\enorm{\D\z}=1, \norm{\z}_{0}\leq s} \frac{\norm{(\D^*\D\z)|_{(\supp \z)}}_1 + \norm{(\D^*\D\z)|_{(\supp \z)^c}}_{1}}{\sqrt{s}}\\ 
&\leq \sup_{\enorm{\D\z}=1, \norm{\z}_{0}\leq s} \norm{\D^*\D\z}_2 + \frac{\norm{(\D^*\D\z)|_{(\supp \z)^c}}_{1}}{\sqrt{s}} = 1 + \sup_{\enorm{\D\z}=1, \norm{\z}_{0}\leq s}\tfrac{1}{\sqrt{s}} \sum_{j\in(\supp \z)^c} |(\D^*\D \z)[j]| \\
&\leq 1 + \sup_{\enorm{\D\z}=1, \norm{\z}_{0}\leq s}\tfrac{1}{\sqrt{s}} (N-s)\tfrac{L\|\z\|_1}{N} \leq 1 + \tfrac{(N-s)L\|\z\|_2}{N}  \leq 1 + L \sqrt{2}.
\end{align*}
\hfill $\square$ 

\noindent{\bf Remark.  }
We note here that the dependence on $L$ seems quite pessimistic.  Indeed, one often instead considers the redundancy of the dictionary, namely $r=N/n$.  In this case, even for redundancy $r=2$ we would have $L=n$ and require a full set of measurements according to this argument.  However, we conjecture that rather than a dependence on $L$ one can reduce the dependence to one on the redundancy $r$.  Intuitively, due to the restriction of Parseval normalization on the frame $\D$, there is a tradeoff between redundancy and coherence; an increase in redundancy should be in some way balanced by a decrease in coherence.  This conjecture is supported by numerical experiments and by theoretical lower bounds for robust recovery in the co-sparse model \cite{GPV14}, but we leave a further investigation to future work.

This example also plays an important role in the so-called \textit{off the grid} compressed sensing setting \cite{tang2012compressive,stoica2012sparse}.  In this framework, the signal frequencies are not assumed to lie on a lattice but instead can assume any values in a continuous interval.  When the signal parameters do not lie on a lattice, the signal may not be truly sparse in the discrete dictionary, and refining the grid to incorporate finer frequencies may lead to numerical instability.  In addition, typical compressed results are difficult to apply under discretization, making off the grid approaches advantageous.  The example we discuss above gives a possible compromise to nearly on the grid recovery from linear measurements.

Another line of research that relates to this example is that of superresolution (see, e.g., \cite{reading2, reading1} and many followup works). These works study the recovery of frequency sparse signals from equispaced samples. No assumptions regarding an underlying grid are made, but rather one assumes a separation distance between the active frequencies. In this sense, the nearly-on-the-grid example just discussed satisfies their assumption as the corresponding separation distance is close to the grid spacing. However, the nature of the resulting guarantees is somewhat different. For example, in these works, every deviation from an exactly sparse signal must be treated as noise, whereas in our result above we have an additional term capturing the compressibility of a signal. We think that an in-depth comparison of the two approaches is an interesting topic for future work.

	\subsection{Example: Fourier measurements and Haar frames of redundancy 2}\label{sec:variable}

In this subsection, we present a second example of a sampling setup that satisfies the assumptions of incoherence and \unrecoverable~ of Theorem~\ref{thm:main}. In contrast to the previous example, one needs to precondition and adjust the sampling density to satisfy these assumptions according to Corollary \ref{thm:variable}, which allows us to understand the setup of the Fourier measurement basis and a 1D Haar wavelet frame with redundancy 2, as introduced in the following.
Let $n = 2^p$. Recall that the univariate discrete Haar wavelet basis of $\C^{2^p}$ consists of $h^0 = 2^{-p/2} (1, 1, \dots, 1),  h=h_{0,0}=2^{-p/2} (1,1,\dots, 1,-1,-1,\dots, -1)$ and the frame basis elements $h_{\ell,k}$ given component wise by
\begin{align*}
  h_{\ell,k}[j] = h[2^{\ell} j -k]
  &= \begin{cases}
              2^{\frac{\ell-p}{2}} \quad&\text{for }\qquad \qquad \qquad  k2^{p-\ell} \leq j< k2^{p-\ell} +2^{p-\ell-1}\\
	      -2^{\frac{\ell-p}{2}} \quad&\text{for }\quad\  k2^{p-\ell} +2^{p-\ell-1}\leq j < k2^{p-\ell} +2^{p-\ell}\\
              0\quad &\text{else,}
	\end{cases}
\end{align*}
for any $(\ell,k)\in\mathbb{Z}^2$ satisfying $0<\ell<p$ and $0\leq k<2^\ell$. The corresponding basis transformation matrix is denoted by $\H$.

One can now create a {\em wavelet frame of redundancy 2} by considering the union of this basis
and a circular shift of it by one index. That is, one adds the vector
$ \tilde{h}^{0} = h^{0} = 2^{-p/2}(1,1,\dots,1)$ and vectors of the form $\tilde{h}_{\ell,k}[j] = h_{\ell,k}[j+1]$ for all 
$(\ell,k)$.  Here we identify $2^p + 1 = 1$.  This is also an orthonormal basis -- its basis transformation matrix will be denoted by $\tH$ in the following, and the matrix $\D \in \mathbb{C}^{2^p \times 2^{p+1}}$ with columns
\begin{eqnarray}
\label{frameHaar}
\D(:,(\ell,2k-1)) &=&  \frac{1}{\sqrt{2}} h_{\ell,k},  \nonumber \\
\D(:,(\ell,2k)) &=& \frac{1}{\sqrt{2}} \tilde{h}_{\ell,k}
\end{eqnarray}
forms a Parseval frame with redundancy 2.  

Corollary \ref{thm:variable} applies to the example where sparsity is with respect to the redundant Haar frame and where sampling measurements are rows $\{ \r_k \}$ from the $n \times n$ orthonormal DFT matrix. Indeed, following Corollary 6.4 of \cite{KW14}, we have the following coherence estimates:
$$
\max_{\ell, j} \{ | \< \r_k, h_{\ell,j} \> |, | \< \r_k, \tilde{h}_{\ell,j} \> | \} \leq \kappa_k := 3 \sqrt{2\pi} / \sqrt{k}.
$$
Since $\| \kappa \|_2^2 = 18\pi \sum_{k=1}^n k^{-1} \leq 18 \pi \log_2(n)$ grows only mildly with $n$, the Fourier / wavelet frame example is a good fit for Corollary \ref{thm:variable}, provided the \unrecoverable~ of the Haar frame is also small.  We will show that the \unrecoverable~ of the Haar frame is bounded by $\unrec \leq  \sqrt{2\log_2(n)}$, leading to the following corollary.

\begin{corollary}
\label{fourier:haar}
Fix a sparsity level $s < N$, and constant $0 < \delta < 1$.  Let $\D \in\C^{n\times N}$ ($N = 2n$) be the redundant Haar frame as defined in \eqref{frameHaar} and let $\A \in \C^{n \times n}$ with rows  $\{\r_1,\dots, \r_n\}$ be the orthonormal DFT matrix.   Consider the diagonal matrix $\W = \operatorname{diag}(\w) \in \C^{n \times n}$ with $w_{k} = C' (\log_2(n))^{-1/2} k^{1/2}$, and construct $\widetilde{\A}\in\C^{m\times n}$ by sampling rows from $\A$ i.i.d. from the probability measure $\nu$ on $[n]$ with power-law decay $\nu(k) = C k^{-1}$.

As long as the number of measurements satisfies $m \geq C_1\delta^{-2} s \log^3(s \log(n)) \log^{3}(n),$
then with probability $1-n^{-\log^3 s}$, the following holds for every signal $\f$: the solution $\f^{\sharp}$ of the weighted $\ell_1$-analysis problem \eqref{eq:l1anaW} with $\y=\widetilde{\A} \f+\e$ for noise $\e$ with weighted error $\| \frac{1}{\sqrt{m}} \W \e \|_2 \leq \varepsilon$ satisfies
\begin{equation}
 \|\f^{\sharp}-\f\|_2 \leq C_2 \varepsilon + C_3 \frac{\|\D^*\f - (\D^*\f)_s\|_1}{\sqrt{s}} \nonumber
\end{equation}
Here $C, C_1, C_2$, and $C_3$ are absolute constants independent of the dimensions and the signal $\f$.
\end{corollary} 

\begin{proof}
To derive this result from Corollary \ref{thm:variable}, it suffices to prove that the redundant Haar frame as defined above has { \unrecoverable~ at most $\unrec \leq  \sqrt{2\log_2(n)}$}.  
We will show that for each $s$-sparse $\z\in\R^N$, $\D^*\D\z$ is at most $3 s \log_2 n$-sparse.  Here, $\D^{*} \D$ is the $N \times N$ Gramian matrix.  For that, consider first either of the two (equal) frame elements $\tfrac{1}{\sqrt{2}}h^0$ and $\tfrac{1}{\sqrt{2}} \tilde{h}^{0}$.  Each of these frame elements has non-zero inner product with exactly two frame elements: $\tfrac{1}{\sqrt{2}}h^0$ and $\tfrac{1}{\sqrt{2}}\tilde{h}^{0}$.  So the corresponding columns of $\D^*\D$ have only two non-vanishing entries.

Consider then a non-constant frame element $\tfrac{1}{\sqrt{2}} h_{\ell,k}$.  Because $\H$ is an orthonormal basis, it is orthogonal to all ${h}_{\ell', k'}$ with $(\ell',k') \neq (\ell,k)$, which is why $(\D^*\D)_{(k,2\ell), (k',2\ell')}=(\D^*\D)_{(k,2\ell-1), (k',2\ell'-1)}=0$ for  $(\ell',k') \neq (\ell,k)$.
So it remains to consider correlations with the $\tilde h_{\ell',k'}$. Again by orthogonality, one has for $(\ell',k')\neq (\ell,k)$
\[
 \langle \tilde{h}_{\ell',k'}, h_{\ell,k}\rangle = \delta_{(\ell',k'), (\ell,k)}+ \langle \tilde{h}_{\ell',k'}, h_{\ell,k}-\tilde{h}_{\ell,k}\rangle,
\]
where $\delta_{(a,b),(c,d)}=\begin{cases}
                             1 & \text{ if }(a,b)=(c,d)\\
                             0 & \text{ else.}
                            \end{cases}
$ denotes the Kronecker delta. By definition, $h_{\ell,k}-\tilde{h}_{\ell,k}$ has only three non-zero entries. On the other hand, for each level $\ell'$, the supports of the $\tilde{h}_{\ell',k'}$ are disjoint for different values of $k'$. So for each $\ell'$, the supports of at most $3$ of the $\tilde{h}_{\ell',k'}$ can intersect the support of $h_{\ell,k}-\tilde{h}_{\ell,k}$. As for $\ell'=\ell$, one of these $\tilde{h}_{\ell',k'}$ must be $\tilde{h}_{\ell,k}$, we conclude that for each $\ell'$ at most $3$ of the $|\langle \tilde{h}_{\ell',k'}, h_{\ell,k}\rangle|$ can be nonzero.
As there are $\log_2 n$ levels $\ell'$, this contributes at most $3\log_2 n$ nonzero entries in the column of $\D^*\D$ indexed by $(k,2\ell)$. Together with $(\D^*\D)_{(k,2\ell), (k,2\ell)}=1$ and noting that a similar analysis holds for columns indexed by $(k,2\ell-1)$, we obtain that each column of $\D^*\D$ has at most $3\log_2(n)+1$ non-zero entries. Now for each $s$-sparse $\z$, $\D^*\D\z$ is a linear combination of the $s$ columns corresponding to $\supp \z$. Consequently, $\|\D^*\D\z-(\D^*\D\z)_s\|_0\leq 3s\log_2(n)$ and thus, by Cauchy-Schwarz and noting that for a Parseval frame $\D$, $\D \D^*$ is the identity,
\begin{align*}
\unrec &= \sup_{\enorm{\D\z}=1, \norm{\z}_{0}\leq s} \frac{\norm{\D^*\D\z}_1 }{\sqrt{s}} =\sup_{\enorm{\D^*\D\z}=1, \norm{\z}_{0}\leq s} \frac{\norm{(\D^*\D\z)_{s}}_1 + \norm{\D^*\D\z - (\D^*\D\z)_{s}}_{1}}{\sqrt{s}}\\ 
&\leq 1 + \sup_{\enorm{\D^*\D\z}=1, \norm{\z}_{0}\leq s} \frac{\sqrt{3s\log_2 n} \norm{\D^*\D\z - (\D^*\D\z)_{s}}_{2}}{\sqrt{s}}\\
&\leq 1 + \sup_{\enorm{\D^*\D\z}=1} \sqrt{3\log_2 n} \norm{\D^*\D\z}_{2} =  1 + \sqrt{3\log_2 n}.
\end{align*}
\end{proof}

\noindent{\bf Remark.}  
 Note that this proof is closely related to the observation that each signal that is $s$-sparse in the Haar frame of redundancy $2$ is also $O(s\log n)$-sparse in the Haar wavelet basis. So a number of measurements comparable to the one derived here also allows for recovery of the wavelet basis coefficients. In addition, however, a more refined analysis suggests that the entries of $\D^*\D\z$ decay quickly - provided not too many approximate cancellations happen. We conjecture that the number of such approximate cancellations can be controlled using a more sophisticated analysis, but we leave this to follow-up work. We hence conjecture that the logarithmic factor can be removed, making the required number of measurements smaller than  for a synthesis approach.

\section{Proof of main result}\label{sec:proof}
Our proof of Theorem \ref{thm:main} extends the analysis in \cite{rauhut2013interpolation}, which extends the analysis in \cite{cheraghchi2013restricted} to weighted sparsity and improves on the analysis in \cite{RefWorks:285}, from orthonormal systems to redundant dictionaries.

\begin{proof}[Proof of Theorem \ref{thm:main}]
Consider the set
$$
\Dsn \defby \{\u : \u\in\C^n, \u = \D\z, \norm{\z}_{\omega,0} \leq s, \enorm{\u} = 1\}.
$$
Let $\Atilde$ be the $m\times n$ subsampled matrix of $\A$ as in the theorem.  Then observe that the smallest value $\delta = \delta_s$ which satisfies the $\D$-$\omega$RIP bound~\eqref{DRIPbound} for $\Atilde$ is precisely
$$
\delta_s = \sup_{\u\in\Dsn} \u^*(\Atilde^*\Atilde-\Id_n)\u.
$$
Since $\Atilde^*\Atilde - \Id_n$ is a self-adjoint operator, we may instead define for any self-adjoint matrix $\B$ the operator
$$
\tripnorm{\B} \defby \sup_{\u\in\Dsn} \ip{\B\u}{\u}.
$$
and equivalently write
$$
\delta_s = \tripnorm{\Atilde^*\Atilde - \Id_n}.
$$
Our goal is thus to bound this quantity.  To that end, let $\rtilde_1, \rtilde_2, \ldots, \rtilde_m \in \C^{n}$ denote the $m$ randomly selected rows of $\A$ that make up $\Atilde$.  It follows from Definition~\ref{BONS} that $\E \rtilde_i^{*} \rtilde_i = \Id_n$ for each $i$, where the expectation is taken with respect to the sampling measure $\nu$.  We thus have 
\begin{equation}\label{delta}
\delta_s = \tripnorm{\Atilde^*\Atilde - \Id_n} = \tripnorm{\frac{1}{m}\sum_{i=1}^m \rtilde_i^{*} \rtilde_i - \Id_n} = \frac{1}{m}\tripnorm{\sum_{i=1}^m \left(\rtilde_i^{*} \rtilde_i - \E\rtilde_i^{*} \rtilde_i\right)}.
\end{equation}
We may bound the moments of this quantity using a symmetrization argument (see, e.g., Lemma 6.7 of~\cite{ra09-1}):
\begin{align}\label{symm}
\E\tripnorm{\sum_{i=1}^m \left(\rtilde_i^{*} \rtilde_i - \E\rtilde_i^{*} \rtilde_i\right)}  &\leq 2\left(\E\tripnorm{\sum_{i=1}^m \epsilon_i\rtilde_i^{*} \rtilde_i} \right) \nonumber \\
&= 2 \E_{\rtilde} \E_{\epsilon} \sup_{\x \in \Dsn} | \langle \sum_{i=1}^m \epsilon_{i}  \rtilde_i^{*} \rtilde_i \x, \x \rangle | \nonumber \\
 &= 2 \E_{\rtilde} \E_{\epsilon} \sup_{\x \in \Dsn} \left| \sum_{i=1}^m \epsilon_{i} | \langle \rtilde_i, \x \rangle |^2 \right|
\end{align}
where $\epsilon_i$ are independent symmetric Bernoulli $(\pm 1)$ random variables.  

Conditional on $(\rtilde_i)$, we have a subGaussian process indexed by $\Dsn$.  For a set $T$, a metric $d$, and a given $t > 0$, the \emph{covering number} $\N(T, d, t)$ is defined as the smallest number of balls of radius $t$ centered at points of $T$ necessary to cover $T$ with respect to $d$.  For fixed $\rtilde_i,$ we work with the (pseudo-)metric
\begin{equation}
\label{pseudo}
d(\x, \z) = \left(  \sum_{i=1}^m ( | \langle \rtilde_i, \x \rangle |^2 - | \langle \rtilde_i, \z \rangle |^2 )^2  \right)^{1/2}
\end{equation}
Then  { Dudley's inequality}~\cite{dudley1967sizes} implies 
 \begin{equation}
 \label{dudley0}
\E_{\epsilon} \sup_{\x \in \Dsn} | \langle \sum_{i=1}^m \epsilon_{i}  \rtilde_i^{*} \rtilde_i \x, \x \rangle | \leq 4 \sqrt{2} \int_{0}^{\infty} \sqrt{\log(\N(\Dsn, d, t))}\,dt.
 \end{equation}
 Continuing as in ~\cite{rauhut2013interpolation}, inspired by the approach of~\cite{cheraghchi2013restricted}, we estimate the metric $d$ using H{\"o}lder's inequality with exponents $p \geq 1$ and $q \geq 1$ satisfying $1/p + 1/q = 1$ to be specified later.  Using also the reverse triangle inequality, we have for $\x, \z \in \Dsn$,
\begin{equation}
\label{continue}
d(\x,\z) \leq 2 \sup_{\u \in \Dsn} \left(   \sum_{i=1}^m | \langle \rtilde_i, \u \rangle |^{2p} \right)^{1/(2p)} \left( \sum_{i=1}^m | \langle \rtilde_i, \x - \z \rangle |^{2q} \right)^{1/(2q)}.
\end{equation}
We will optimize over the values of $p,q$ later on.  To further bound this quantity, we have $ |\ip{\rtilde_i}{\d_j}| \leq \omega_j$ by assumption.   Recall the \unrecoverable~ in Definition \ref{eq:Wunrec}.
Then, fixing   $\u\in\Dsn$ such that $\u=\D\z$ with $\norm{\z}_{\omega,0} \leq s$, we have for any realization of $(\rtilde_i)$,
	\begin{align}\label{old88}
 |\ip{\rtilde_i}{\u}| &=  |\ip{\D\D^*\rtilde_i}{\u}|\notag\\
	&=  |\ip{\D^*\rtilde_i}{\D^*\D\z}|\notag\\
	&\leq  \sum_j |\ip{\rtilde_i}{\d_j}|_j |(\D^{*} \D \z)_j |   \notag \\
		&\leq   \sum_j \omega_j |(\D^{*} \D \z)_j | \notag \\
		&= \norm{\D^*\D\z}_{\omega,1} \notag \\
	&\leq  \sqrt{s} \unrec.
	\end{align}
	The first line uses that $\D$ is a Parseval frame, while the last line uses the definition of $\unrec$.  The 
	quantity $s \unrec^2$ will serve as a rescaled sparsity parameter throughout the remaining proof.
		
Continuing from \eqref{continue}, we may bound
\begin{align}
\label{bound_gen}
	\sup_{\u \in \Dsn} \left( \sum_{i=1}^m | \langle \rtilde_i, \u \rangle |^{2p} \right)^{1/(2p)} &= 
	\sup_{\u \in \Dsn} \left( \sum_{i=1}^m | \langle \rtilde_i, \u \rangle |^{2}| \langle \rtilde_i, \u \rangle |^{2p-2} \right)^{1/(2p)} \notag \\
	&\leq (s \unrec^2)^{(p-1)/(2p)} \left( \sup_{\u \in \Dsn} \sum_{i=1}^m | \langle \rtilde_i, \u \rangle |^2  \right)^{1/(2p)}.
\end{align}
We now introduce the (semi-)norm
$$
\| \u \|_{X,q} := \left( \sum_{i=1}^m | \langle \rtilde_i, \u \rangle |^{2q} \right)^{1/(2q)}.
$$
Using basic properties of covering numbers and the bound in \eqref{dudley0}, we obtain
\begin{align}
\label{dudley2}
\mathbb{E}_{\epsilon}  \sup_{\u \in \Dsn} | \langle \sum_{i=1}^m \epsilon_{i}  \rtilde_i^{*} \rtilde_i \u, \u \rangle | &\leq C_1 (s \unrec^2)^{(p-1)/(2p)}  \left( \sup_{\u \in \Dsn} \sum_{i=1}^m | \langle \rtilde_i, \u \rangle |^2  \right)^{1/(2p)} \int_{0}^{\infty}  \sqrt{\log(\N(\Dsn, \| \cdot \|_{X,q}, t))}\,dt
\end{align}
where $C_1$ is an absolute constant.

We now estimate the key integral above in two different ways, the first bound being better for small values of $t$, the second bound better for larger values of $t$.

\begin{description}
\item[Small values of $t$: ] Following \eqref{old88} and \eqref{bound_gen} and since $\|\u\|_2=1$ for all $\u\in\Dsn$, we have that for any $\u \in \Dsn$,
\begin{equation}
\label{simple_bound}
\| \u \|_{X,q} \leq  (s \unrec^2)^{1/2} m^{1/(2q)}.
\end{equation}
Standard covering arguments show that for any seminorm $\norm{\cdot}$, one has for the unit ball $B_{\norm{\cdot}}\subset \C^s$ that $\N(B_{\norm{\cdot}}, \norm{\cdot}, t) \leq (1+1/t)^{2s}$.  Thus for $\Lambda\subset\{1, \dots, N\}$ with $|\Lambda|=s'$, one can define the seminorm $\norm{\cdot}_{\Lambda}$ on $\C^{s'}$ by 
	$$\norm{\y}_{\Lambda} \defby \left( \sum_{i=1}^m | \ip{\D_\Lambda^*\rtilde_i^{*}}{\y} |^{2q} \right)^{1/(2q)},$$
	 and observe that for $\u$ with $\| \u \|_{X,q} \leq 1$ and $\u=\D_\Lambda\y$ for $\y\in\C^{s'}$ that one has the equivalence $ \norm{\y}_{\Lambda} = \norm{\u}_{X,q}$. Applying this equivalence for $\Lambda=\supp(\z)$ and $\y=\z_\Lambda$, we have that
	\begin{align*}
	\N(\Dsn, \| \cdot \|_{X,q}, t) &\leq \sum_{\omega(\Lambda)\leq s} \N(B_{\|\cdot\|_\Lambda}, m^{1/(2q)} \| \cdot \|_2 \sqrt{(s \unrec^2)},t )  \\
	&\leq {\binom{N}{s}}(1 + \sqrt{(s \unrec^2)} m^{1/(2q)} t^{-1})^{2s}\\
	&\leq (eN/s)^{s}(1 + \sqrt{(s \unrec^2)} m^{1/(2q)}t^{-1})^{2s},
	\end{align*}
	where we have applied \cite[Proposition C.3]{RefWorks:45}   in the final line.

	\item[Large values of $t$: ] To get a bound for larger values of $t$, we will utilize another embedding, and for that reason we define the auxiliary set
	$$
	\DsnOne \defby \{\u : \norm{\D^*\u}_{\omega,1} \leq 1, \u=\D\z \in \C^n, \| \z \|_{\omega,0} \leq s\} = \bigcup_{\Lambda: \omega(\Lambda) \leq s}\BoneS,
	$$
	where $\BoneS \defby \{\u : \norm{\D^*\u}_{\omega,1} \leq 1, \u=\D\z \in \C^n, \supp(\z) = \Lambda\}$.
	Then we have the embedding $\Dsn \subset \sqrt{(s \unrec^2)}\DsnOne$ since for any $\u\in\Dsn$, 
 $\norm{\D^*\u}_{\omega,1} \leq \sqrt{(s \unrec^2)}$ by definition of the localization factor $\unrec$.

	We now use a variant of Maurey's lemma \cite{maurey}, precisely, the variant as stated in Lemma 5.3 in \cite{rauhut2013interpolation}, in order to deduce a different covering number bound:
	\begin{lemma}[Maurey's lemma]
	\label{maurey}
	For a normed space $X$, consider a finite set $U \subset X$ of cardinality $N$, and assume that for every $L \in \mathbb{N}$ and $(\u_1, \dots, \u_L) \in U^L$, $\mathbb{E}_{\epsilon} \| \sum_{j=1}^L \epsilon_j \u_j \|_X \leq A \sqrt{L}$, where $\epsilon$ denotes a Rademacher vector.  Then for every $t > 0$,
	$$
	\log{\N\left(\emph{conv}(U), \| \cdot \|_X, t \right)} \leq c (A/t)^2 \log(N),
	$$
	where $\text{conv}(U)$ is the convex hull of $U$.
	\end{lemma}

As a corollary, we have the following.
	\begin{corollary}\label{lem:covering}
For every $t > 0$,
	$$
	\log(\N(\Dsn, \| \cdot \|_{X,q}, t) \leq C \left( 2 e^{-1/2} \sqrt{2q} m^{1/(2q)} \sqrt{2(s \unrec^2)} t^{-1} \right)^2 \log(N).
	$$
	\end{corollary}
	\begin{proof}[Proof of Corollary \ref{lem:covering}]
	Consider Lemma \ref{maurey} with norm $\| \cdot \|_{X,q}$, parameter $A = 2 e^{-1/2} \sqrt{2q} m^{1/(2q)}$, and 
	$U = \{ \pm \omega_j^{-1} \D^{*}(\e_j), \pm i \omega_j^{-1} \D^*(\e_j), j \in [N] \}$.  
First observe that 
$$\Dsn \subset  \sqrt{(s \unrec^2)}\DsnOne \subset \sqrt{2(s \unrec^2)} \text{conv}(U).$$
For a Bernoulli random vector $\epsilon = (\epsilon_1, \dots, \epsilon_L)$ and $\u_1, \dots, \u_L \in U$ we have
\begin{align}
\mathbb{E}_{\epsilon} \| \sum_{j=1}^L \epsilon_j \u_j \|_{X,q} &\leq \left( \mathbb{E} \| \sum_{j=1}^L \epsilon_j \u_j \|_{X,q}^{2q} \right)^{1/(2q)} \notag \\
&= \left( \mathbb{E} \sum_{\ell=1}^m | \langle{\rtilde_{\ell}, \sum_{j=1}^L \epsilon_j \u_j \rangle} |^{2q}  \right)^{1/(2q)}
\notag \\
&=  \left(  \sum_{\ell=1}^m \mathbb{E} | \langle{\rtilde_{\ell}, \sum_{j=1}^L \epsilon_j \u_j \rangle} |^{2q}  \right)^{1/(2q)}
\notag \\
&\leq 2 e^{-1/2} \sqrt{2q} \left(  \sum_{\ell=1}^m \| ( \langle{\rtilde_{\ell},\u_j \rangle} )_{j=1}^L \|_2^{2q} \right)^{1/(2q)},
\end{align}
where we applied Khintchine's inequality \cite{khintchine1923dyadische} in the last step.  
Since each $\u_j$ consists of a multiple of a single column of $\D$, we also have for each $j$ and $i$ that $|\ip{\rtilde_i}{\u_j}| = |\ip{\D^*\rtilde_i^{*}}{\omega_j^{-1} \e_{j}}|$ for some coordinate vector $\e_{j}$ consisting of all zeros with a $1$ in the $j$th position.  Thus $|\ip{\rtilde_i}{\u_j}| \leq 1$ for each $j$ and $i$, which means that 
$\enorm{(\ip{\rtilde_i}{\u_j})_{k=1}^L} \leq \sqrt{L}$ for any $L$ and 
$$
\mathbb{E}_{\epsilon} \| \sum_{j=1}^L \epsilon_j \u_j \|_{X,q} \leq 2e^{-1/2}\sqrt{2q}m^{1/(2q)} \sqrt{L} = A \sqrt{L}.
$$
The corollary then follows from Lemma \ref{maurey}.
\end{proof}

\end{description}
	
We have thus obtained the two bounds for $t > 0$:
	\begin{align*}
	\sqrt{\log(\N(\Dsn, \Xnorm{\cdot}, t))} &\leq  \sqrt{ s \log(eN/s) + 2s \log\left( 1 + 2\sqrt{(s \unrec^2)}m^{1/(2q)}t^{-1} \right)} \nonumber \\
	\sqrt{\log(\N(\Dsn, \Xnorm{\cdot}, t))} &\leq C \sqrt{q} m^{1/(2q)} \sqrt{(s \unrec^2)} t^{-1} \sqrt{\log(N)}.
	\end{align*}	
	We may now bound the integral in~\eqref{dudley0}.   Without loss, we take the upper integration bound as $t_0 = \sqrt{(s \unrec^2)} m^{1/(2q)}$ because, for $t >t_0$, we have $\N(\Dsn, \| \cdot \|_{X,q}, t) = 1$ by \eqref{simple_bound}.  Splitting the integral in two parts and using our covering number bounds, we have for $\alpha \in (0, t_0)$,
	\begin{align*}
	\int_{0}^{\infty} &\sqrt{\log(\N(\Dsn, \| \cdot \|_{X,q}, t))}\,dt\\
	&\leq \int_{0}^{\alpha}  \sqrt{ s \log(eN/s) + 2s \log\left( 1 + 2\sqrt{(s \unrec^2)}m^{1/(2q)}t^{-1} \right)} dt \hspace{.5mm} + C_1 \sqrt{q} m^{1/(2q)} \sqrt{(s \unrec^2)}  \int_{\alpha}^{t_0} t^{-1} dt  \nonumber \\
	&\leq \alpha \sqrt{s \log(eN/s)} + \sqrt{2s}\alpha \sqrt{ \log(e(1 + \sqrt{(s \unrec^2)} m^{1/(2q)})\alpha^{-1}} + C_2 \sqrt{q m^{1/q} (s \unrec^2) \ln (4N)} \log(\sqrt{(s \unrec^2)} m^{1/(2q)}/\alpha)
	\end{align*}
	where in the last step we have utilized that for $a>0$, $\int_0^a \sqrt{\ln(1+t^{-1})}\,dt \leq a\sqrt{\ln(e(1+a^{-1}))}$ (see Lemma 10.3 of~\cite{ra09-1}).
	{ Choosing $\alpha = m^{1/(2q)}$ } yields 
	$$
	\int_{0}^{\infty} \sqrt{\log(\N(\Dsn, \| \cdot \|_{X,q}, t))}\,dt \leq C_3 \sqrt{q (s \unrec^2) m^{1/q} \log(N) \log^2(s \unrec^2)}.
	$$	
	Combining this with~\eqref{delta},~\eqref{symm} and \eqref{dudley2},we have
	\begin{align*}
	\E \delta_s &\leq \frac{C_3 (s \unrec^2)^{(p-1)/(2p)} \sqrt{q m^{1/q} (s \unrec^2)  \log(N) \log^2(s \unrec^2)}}{m} \E \sup_{\x \in \Dsn} \left( \sum_{i=1}^m | \langle \rtilde_i, \x \rangle |^2 \right)^{1/(2p)} \nonumber \\
	&\leq \frac{C_3 (s \unrec^2)^{1/2 + (p-1)/(2p)} \sqrt{q \log(N) \log^2(s \unrec^2)}}{m^{1 - 1/(2q)}m^{-1/(2p)}} 
	\E \left( \frac{1}{m} \tripnorm{ \sum_{i=1}^m \rtilde_i^{*} \rtilde_i - \Id_n} + \tripnorm{ \Id_n } \right)^{1/(2p)} \nonumber \\
	&\leq  \frac{C_3 (s \unrec^2)^{1/2 + (p-1)/(2p)} \sqrt{q \log(N) \log^2(s \unrec^2)}}{m^{1/2}}  \sqrt{ \E \delta_s + 1}.
	\end{align*}
Above we applied H{\"o}lder's inequality and used that $1/q + 1/p = 1$ as well as $p \geq 1$.  { We now choose $p = 1 + 1/\log(s \unrec^2)$ and $q = 1 + \log(s \unrec^2)$ to give} $(s \unrec^2)^{(p-1)/(2p)} \leq (s \unrec^2)^{(p-1)/2} = (s \unrec^2)^{1/(2\log(s \unrec^2))} = e^{1/2}$ and
$$
\E \delta_s \leq C_4 \sqrt{2 \log(N) \log^2(s \unrec^2)/m} \sqrt{ \E \delta_s + 1}.
$$		
	Squaring this inequality and completing the square finally shows that
	\begin{equation}\label{prob}
	\E \delta_s \leq C_5 \sqrt{ \frac{(s \unrec^2) \log(N) \log^3(s \unrec^2)}{m}},
	\end{equation}
provided the term under the square root is at most 1.  Then $\E \delta_s \leq \delta/2$ for some $\delta \in (0,1)$ if
\begin{equation}
\label{end_m}
m \geq C_6 \delta^{-2} (s \unrec^2) \log^3(s \unrec^2) \log(N).
\end{equation}
\medskip

It remains to show that $\delta_s$ does not deviate much from its expectation.  For this \emph{probability bound}, we may write $\delta_s$ as the supremum of an empirical process as in [\cite{ra09-1}, Theorem 6.25] and apply the following Bernstein inequality for the supremum of an empirical process:

\begin{theorem}[\cite{B03}, Theorem 6.25 of \cite{ra09-1}]
Let ${\mathcal{F}}$ be a countable set of functions $f : \C^n \rightarrow \R$.  Let $Y_1, \dots, Y_m$ be independent copies of a random vector $Y$ on $\C^n$ such that $\E f(Y) = 0$ for all $f \in {\mathcal{F}}$, and assume $f(Y) \leq 1$ almost surely.  Let $Z$ be the random variable $Z = \sup_{f \in {\mathcal{F}}} \sum_{\ell=1}^m f(Y_{\ell})$, and $\E Z$ its expectation.  Let $\sigma^2 > 0$ such that $\mathbb{E} [ f(Y)^2 ] \leq \sigma^2$ for all $f \in {\mathcal{F}}$.  Then, for all $t > 0$,
\begin{equation}
\mathbb{P}(Z \geq \E Z + t) \leq \exp\left( - \frac{t^2}{2(m \sigma^2 + 2 \E Z) + 2t/3} \right)
\end{equation}
\end{theorem}

We apply this theorem to provide a probability bound.  Let $f_{z,w}(\rtilde) = \text{Re}(\langle (\rtilde_i^{*} \rtilde_i - \Id)z, w \rangle )$ so that
$$
m \delta_s = \tripnorm{\sum_{i=1}^m \left(\rtilde_i^{*} \rtilde_i - \E\rtilde_i^{*} \rtilde_i\right)}  =  \sup_{(z,w) \in {\mathcal{D}}, } \sum_{i=1}^m f_{z,w}(\rtilde_{i}).
$$
Since $\E \rtilde_i^{*} \rtilde_i = \Id$ we have $\E f_{z,w}(\rtilde) = 0$.  Moreover, 
$| f_{z,w}(\rtilde) | \leq \max_{z \in {\mathcal{D}}} |\langle \rtilde_i, z \rangle|^2 + 1 \leq   s \eta^2 + 1$.  For the variance term, 
we have $\E | f_{z,w}(\rtilde_i)|^2 = \E \| (\rtilde_i^{*} \rtilde_i  - \Id) z \|_2^2  \leq (s \eta^2 + 1)^2.$ 

\bigskip

Fix $\delta \in (0,1)$, and suppose the number of measurements $m$ satisfies \eqref{end_m} so that $\E \delta_s \leq \delta/2$.  Then it follows
	\begin{eqnarray}\label{tail}
\P(\delta_s \geq \delta) &\leq& \P(\delta_s \geq \E \delta_s + \delta/9) \nonumber \\
&=&  \P \left(\tripnorm{\sum_{i=1}^m \left(\rtilde_i^{*} \rtilde_i - \E\rtilde_i^{*} \rtilde_i\right)}  \geq \E \tripnorm{\sum_{i=1}^m \left(\rtilde_i^{*} \rtilde_i - \E\rtilde_i^{*} \rtilde_i\right)} + \delta m/9 \right)  \nonumber \\
&=&   \P \left(\frac{1}{s \eta^2 + 1} \tripnorm{\sum_{i=1}^m \left(\rtilde_i^{*} \rtilde_i - \E\rtilde_i^{*} \rtilde_i\right)}  \geq \frac{1}{s \eta^2 + 1} \E \tripnorm{\sum_{i=1}^m \left(\rtilde_i^{*} \rtilde_i - \E\rtilde_i^{*} \rtilde_i\right)} + \frac{\delta m/9}{s \eta^2 + 1}  \right)  \nonumber \\
&\leq&  \exp \left(- \frac{(\frac{\delta m/9}{s \eta^2 + 1})^2}{2m(1 + \frac{\delta}{s \eta^2+1})  +  \frac{2}{3}(\frac{\delta m/9}{s \eta^2 + 1})} \right)     \nonumber \\
&\leq& \exp{\left( -\frac{\delta^2 m}{C_7 (s \unrec^2)}\right)}.
	\end{eqnarray}
The last term is bounded by $\gamma \in (0,1)$ if $m \geq C \delta^{-2} (s \unrec^2) \log(1/\gamma)$.  Together, we have $\delta_s \leq \delta$ with probability at least $1-\gamma$ if
$$
m \geq C_8 \delta^{-2} s \unrec^2 \max\left\{ \log^3(s\unrec^2) \log(N), \log(1/\gamma) \right\}.
$$
This completes the proof.
	\end{proof}

\section{Conclusion}\label{sec:disc}
We have introduced a coherence-based analysis of compressive sensing when the signal to be recovered is approximately sparse in a redundant dictionary.  Whereas previous theory only allowed for unstructured random sensing measurements, our coherence-based analysis extends to structured sensing measurements such as subsampled uniformly bounded bases, bringing the theory closer to the setting of practical applications.  We also extend the theory of variable density sampling to the dictionary setting, permitting some coherence between sensing measurements and sparsity dictionary.  We further extend the analysis to allow for weighted sparse expansions.  Still, several open questions remain.  While we provided two concrete examples of dictionaries satisfying the bounded \unrecoverable~ condition required by our analysis -- the oversampled DFT frame and redundant Haar wavelet frame -- these bounds can almost certainly be extended to more general classes of dictionaries, and improved considerably in the case of the oversampled DFT frame.  We have also left several open problems related to the full analysis for variable density sampling in this setting, including the removal of a weighted noise assumption in the $\ell_1$-analysis reconstruction method.  Finally, we believe that the $\D$-RIP assumption used throughout our analysis can be relaxed, and that a RIPless analysis \cite{RefWorks:595} should be possible and permit non-uniform signal recovery bounds at a reduced number of measurements.  It would also be useful to extend our results to measurement matrices constructed in a deterministic fashion for those applications in which randomness is not admissable; of course this is a challenge even in the classical setting \cite{RefWorks:36}.

\section*{Acknowledgments}
The authors thank the reviewers of this manuscript for their helpful suggestions which significantly improved the paper.
The authors would also like to thank the Institute for Computational and Experimental Research in Mathematics (ICERM) for its hospitality during a stay where this work was initiated.  In addition, Krahmer was supported by the German Science Foundation (DFG) in the context of the Emmy Noether Junior Research Group KR 4512/1-1 (RaSenQuaSI). Needell was partially supported by NSF CAREER $\#1348721$ and the Alfred P. Sloan Foundation. Ward was partially supported by an NSF CAREER award, DOD-Navy grant N00014-12-1-0743, and an AFOSR Young Investigator Program award.



\bibliography{bib}
\bibliographystyle{myalpha}

\end{document}